\documentclass[12pt]{article}

\usepackage{algorithm,algorithmic}
\usepackage{graphicx}
\usepackage{amsmath}
\usepackage{fullpage}
\usepackage{amssymb}
\usepackage{amsfonts}
\usepackage{latexsym}
\usepackage{amsthm}

\newcommand{\hf}{\frac12}

\renewcommand{\vec}[1]{\ensuremath{\mathbf{#1}}}

\newcommand{\nn}{\vec{n}}
\renewcommand{\div}{\nabla\cdot\,}
\newcommand{\grad}{\ensuremath {\vec \nabla}}

\newcommand{\bbE}{\mathbb{E}}

\newtheorem{thm}{Theorem}

\newtheorem{lem}[thm]{Lemma}

\begin{document}

\title{Simultaneous Source for non-uniform data variance and missing data}

\author{{E.\ Haber\thanks{Department of Mathematics
\& Department of Earth and Ocean Science, The University of British Columbia, Vancouver, BC, Canada, V6T-1Z4,  Phone: +1 (604) 822-9068, Fax: +1 (604) 822-2545}}  \ and {M. Chung\thanks{
Department of Mathematics, Virginia Tech
474 McBryde, Stanger Street, Blacksburg, VA 24061, USA,
Phone: +1 (540) 231-3446     Fax : +1 (540) 231-5960 }}}

\maketitle

\begin{abstract}
The use of simultaneous sources in geophysical inverse problems has revolutionized the ability to deal with large scale data sets that are obtained from multiple source experiments. However, the technique breaks when the data has non-uniform standard deviation or when some data are missing. In this paper we develop, study, and compare a number of techniques that enable to utilize advantages of the simultaneous source framework for these cases. We show that the inverse problem can still be solved efficiently by using these new techniques. We demonstrate our new approaches on the Direct Current Resistivity inverse problem.
\end{abstract}

{\bf keywords:} simultaneous sources, inverse problems, stochastic programming, Direct Current Resistivity\\

\section{Introduction}

In this paper we investigate efficient methods for parameter estimation in partial differential equations (PDEs) with multiple right hand sides. The goal is to infer quantitative information (physical properties) of a media from indirect measurements. Our main focus are geophysical experiments, where a combination of sources and receivers are used. However, application such as electromagnetic imaging and electrical impedance tomography facing similar challenges \cite{smvoz,brocea,devony,na1,jc1}. Sources produce signals that are omitted into a media to be investigated. Receivers collect (measure) data obtained by signal probing the media. If the interaction between the media, sources and receivers can be described by a linear PDE and $n_{s}$ sources are used with $n_{r}$ receivers then the measured data is an $n_{s} \times n_r$ matrix $D$, that can be approximated by
\begin{equation}
\label{eq:for}
D = C \odot (P^{\top}A(u)^{-1}Q + {\cal E}).
\end{equation}
Here, $A(u)$ is an $n \times n$ matrix that is obtained by a discretization of a differential operator that depends on parameters $u$, i.e., the media's properties. The matrix $Q$ of size $n \times n_{s}$ represents a discretization of the sources, and $P$ is an $n \times n_{r}$ matrix that represents the discretization of the receivers. Furthermore $\cal E$ is an $n_{r} \times n_{s}$ random matrix of the noise, assumed to be independent  and normal distributed with zero mean, but not necessarily with the same standard deviation. In Equation~\eqref{eq:for} the operator $\odot$ denotes the Hadamard product. The $n_{r} \times n_{s}$ matrix $C$ typically has three forms.
\begin{itemize}
\item[(a)] First, $C = \sigma^{-1} E$,  $E=e_{n_r} e_{n_s}^{\top}$ is a matrix that consists of entries that are all ones -- all receivers share all sources -- multiplied by the inverse standard deviation $\sigma^{-1}$ which is identical for all data.

\item[(b)] Second, $C$ is a matrix of non-negative numbers $c_{ij} = \sigma_{ij}^{-1} \geq 0$. In this case all receivers share all sources but each datum has a different standard deviation.

\item[(c)] Third, $C$ is a sparse matrix. In this case some receivers share some of the sources but the data is incomplete, that is, not all source-receiver combination are recorded. Clearly, a value $c_{ij}=0$
in the matrix $C$ is a special case of~(b), where some data has infinite standard deviation.
\end{itemize}
We refer to $C$ as the ``variance matrix''. Notice that $C$ is not the covariance matrix but it carries the entries of the diagonal of the covariance matrix (that is the variance). The entries of data matrix, $D_{ij}$, correspond to the $i^{\rm th}$ receiver and the $j^{\rm th}$ source normalized by the standard deviation of each datum.

Various other structures of the matrix $C$ may occur and are application dependent.  While the above cases~(a)--(c) seems to be specific, many applied problems fall into the above category, applications such as seismic imaging, electromagnetic imaging, electrical impedance tomography, diffraction tomography \cite{smvoz,brocea,devony,na1,jc1} and more.

Consider some measured data $D$. Our goal is to obtain a ``reasonable'' estimate of the parameters $u$ that approximately give rise to the observed data. To obtain such an estimate we consider the output least squares approach. That is, we consider the regularized optimization problem
\begin{equation}
\label{eq:opt}
\min_{u} \ J(u) = \hf \|C \odot (P^{\top}A(u)^{-1}Q - D)\|_{\rm F}^{2} + \alpha S(u),
\end{equation}
where $S(u)$ is some regularization term that is used to obtain an estimate of $u$ and incorporates some a-priori information such as smoothness or sparsity. Here, $\alpha$ is a regularization parameter and $\| \, \cdot \, \|_{\rm F}$ denotes the Frobenius norm. The choice of the Frobenius norm rises naturally from a statistical point of view, where the minimization problem can be interpreted as a maximum likelihood approach with a-priori information.

%In order to use fast gradient based optimization methods one need to calculate the derivative of $J$.

It is straight forward to calculate the gradient of $J$ with respect to $u$ (see \cite{hao} for details). In a nutshell, let $Q_{j}$, $C_{j}$ and $D_{j}$ be the $j^{\rm th}$ columns of $Q$, $C$ and $D$ respectively,
$$ Y_{j} = A(u)^{-1}Q_{j} $$
the solution of the forward problem for the $j^{\rm th}$ source, and let
$$ G(u,v) = \grad_{u} (A(u)v) $$
be a sparse matrix containing the derivatives of $A(u)v$ with respect to $u$ for a fixed vector $v$. The gradient of $J$ is given by
\begin{equation}
\label{eq:dopt}
\grad_{u}  J(u) = -\sum_{j=1}^{n_{s}} G(u,Y_{j})^{\top}A(u)^{-\top} P\left((C_{j} \odot C_{j}) \odot (P^{\top}A(u)^{-1}Q_{j} - D_{j}) \right) + \alpha \grad_{u}S(u).
\end{equation}

In order to compute the objective function~\eqref{eq:opt} and its gradient~\eqref{eq:dopt} we essentially need to solve the following two linear systems
$$ A(u)Y = Q \quad {\rm and} \quad A(u)^{\top} \Lambda = P\left((C \odot C) \odot (P^{\top}Y - D) \right).$$
The first problem (the forward) has $n_{s}$ right hand sides while the second (the adjoint) has $n_r$ right hand sides. For small scale problems, this computation can be easily done using either a Cholesky (if the system is symmetric) or a LU factorization (for nonsymmetric systems). However, for large scale problems this computation becomes much more challenging. Indeed, solving large scale linear systems with multiple right hand sides is a difficult task that requires both efficient iterative solvers and large amount of memory. Even if powerful solvers such as multigrid methods can be utilized, the effective cost in computational speed and memory still increases linearly with the number of sources. For problems with many right hand sides such computation can become prohibitively expensive, especially for problems that stem from 3D applications.

\bigskip

A recent approach to solve such a problem that works well in the case that $C=\sigma^{-1}E$, a matrix with all entries $\sigma^{-1}$ -- thus can be neglected ($\sigma^2$ is factored into the regularization parameter) -- is obtained using the identity
\begin{align*}
	 \|P^{\top}A(u)^{-1}Q - D\|_{\rm F}^{2} &= {\sf trace}\left((P^{\top}A(u)^{-1}Q - D)^{\top}(P^{\top}A(u)^{-1}Q - D)\right) \\
	&= \bbE_{w} \left( \|P^{\top}A(u)^{-1}Qw - Dw\|^{2} \right),
\end{align*}
where $\| \, \cdot \, \|$ denotes the Euclidian norm, $\bbE$ is the expected value, and $w$ is a random variable with zero mean and standard deviation of $1$, see \cite{HaberChungHerrmann2011}. Using this identity, problem~\eqref{eq:opt} turns into a stochastic programming problem of the form
\begin{equation}
\label{stopro}
\min_{u} \ J(u) = \bbE_{w}\left(\hf \|P^{\top}A(u)^{-1}Qw - Dw\|^{2}\right) + \alpha S(u).
\end{equation}
The key observation when solving the problem \eqref{stopro} is that {\em given a single realization $w$, the evaluation of the objective function and the gradient, that correspond to that sample requires only a {\bf single} inversion of the matrix $A(u)$}.
Stochastic programming techniques attempt to use as little samples as possible to approximate minimal expected values. It is straight forward to apply efficient stochastic programing algorithms to solve problem~\eqref{stopro}, that requires significantly less inversions of the matrix $A(u)$ than standard methods mentioned above.
This approach, known also as {\bf``simultaneous source''}, was first suggested in the context of waveform inversion \cite{RohmbergGeop2010,krebs09ffw,neelamani08dos,LeeuwenAravkinHerrmann2011} and then explained and analyzed by means of stochastic programming in \cite{HaberChungHerrmann2011}. It has been shown to be highly effective in dealing with multiple sources.

Nonetheless, when $C \not= \sigma^{-1}E$ the above approach loose its advantages and other techniques are required. This is due to the fact that the Hadamard product does not commute with the matrix vector product and thus computing $(C \odot (P^{\top}A(u)^{-1}Q - D))w$ requires first to compute $C \odot (P^{\top}A(u)^{-1}Q - D)$ and only then apply the result to $w$. It is therefore impossible to linearly combine sources and the approach  breaks. One approach that has been used in the literature mentioned above is to ignore the non-uniform standard deviation and to work with $C=E$. This can be seen as a weighted least square approach. If the data has similar standard deviations this approach works well, however, when the data has standard deviations that range over a large dynamical spectrum the approach tend to produce unsatisfactory results. This is because small data values are basically ignored while large data values are fitted. In research areas such as geophysical applications small data values are crucial. Roughly speaking, large data values rise close to the source and thus do not contain much information about the media. On the other hand, small magnitude data contain more information about the media. Ignoring  this data may lead to serious artifacts in the reconstruction (see \cite{parker,taran} and our numerical experiments).

The goal of this paper is to explore techniques that give rise to efficient computations for problem when $C$ is not just of the form $\sigma^{-1} E$. We discuss a number of possible approximations and show that it is possible to generate efficient algorithms for these cases as well. We use the Direct Current (DC) Resistivity example to demonstrate the effectiveness of our approach.

The paper is structured as follows. In Section~\ref{sec2} we present possible extensions of the simultaneous sources approach and discuss their properties. In Section~\ref{sec3} we discuss efficient computations of these techniques. In particular we explore the use of the Sample Average Approximation (SAA) as a stochastic optimization technique. In Section~\ref{sec4} we investigate the numerical efficiency of the different approaches and we summarize our work in Section~\ref{sec5}.

\section{Numerical treatment of the non-standard form}
\label{sec2}

In the previous section we reviewed the case where $C = \sigma^{-1} E$ in which the original deterministic problem is converted into a stochastic programming problem with the feature that only a single matrix inversion is needed for every realization of the random vector. In this section we discuss a number of techniques that enable us to work with a matrix $C$ that does not all contain $\sigma^{-1}$'s. We discuss the different conditions that enable each technique to work.

\subsection{Data completion}
\label{meth1}
Let us assume that the matrix $C$ contains only entries of $\sigma^{-1}$ and $0$'s. Here, entries $c_{ij} = 0$ mirror the fact that the $i^{\rm th}$ receiver is not connected to the $j^{th}$ source and therefore missing and has infinite uncertainty.
The idea is to treat this problem as a missing data problem, where the complete data is given by
$$ D^{\rm all} =  P^{\top}A(u)^{-1}Q + {\cal E} $$
and the observed data is connected to the complete data by
$$ D^{\rm obs} = C \odot D^{\rm all}, $$
where $C$ is a sparse matrix.

It is possible to estimate missing data by matrix completion techniques. Interpolation of data is a major area of research (see \cite{TradUlrychSacchi2001} and reference within). For instance, in seismology, efficient methods for data interpolation are based on sparse recovery. Many other fields have a data-specific approach. Here, we present a simple data interpolation technique based on the solution of a reduced forward problem.

\bigskip

In some cases it is possible to compute the solution to the forward problem at reduced cost.
For example, if $u$ is constant or, in many geophysical applications, admits a 1D or 2D structure
it is possible to compute the solution to the forward problem in a low cost  \cite{wardhow}.
Here we assume that it is possible to solve the 1D or 2D problems and that the solution is readily available. Let the reduced forward problem be
$$ D^{\rm red} = P^{\top} A_{\rm red}(u_{\rm red})^{-1}Q. $$
Here, $u_{\rm red}$ is the reduced (say 1D or 2D)
model and  $A_{\rm red}(u_{\rm red})$ is a reduced forward
model (that is the 1D or 2D forward modeling matrix). The reduced data $D^{\rm red}$
therefore corresponds to the simple 1 or 2D model and we use this data to replace
the missing data.
Thus, we set
$$ \widehat D^{\rm all} = C \odot D^{\rm obs} + (\sigma^{-1} E-C) \odot D^{\rm red}, $$
 We then use  $\widehat D^{\rm all}$ to solve the optimization problem
\begin{equation} \label{stoproComp}
 \min_{u} \ J(u) =  \hf\, \bbE_{w} \left( \|P^{\top}A(u)^{-1}Qw - \widehat D^{\rm all}w\|^{2} \right) + \alpha S(u),
\end{equation}
that can be done by using  the algorithms proposed in~\cite{HaberChungHerrmann2011}.

Completing and using the data in this way has a distinctive advantage. Assuming for a moment, that we do not have any measured data. Then, the corresponding completed data is simply $D^{\rm red}$ which has a solution $u_{\rm red}$. That is, this data completion can be seen as a form of regularization, pulling the solution to a simple 1D or 2D ``reasonable'' model.

There is one obvious disadvantage to this approach. The matrix $\widehat D^{\rm all}$ contains some ``invented'' data. The missing data can be far from the true data and thus significantly  bias the solution. Nonetheless, the approach is straight forward and as we see in numerical experiments in Section~\ref{sec4} may lead to very reasonable results.

\subsection{The case of low rank $C$}
\label{meth2}
In some cases it is possible to decompose $C$ into low rank matrices, or at least, to approximate $C$ by a matrix of small rank. Consider the case
\begin{equation}
\label{CLR}
C = XZ^{\top},
\end{equation}
where $X$ and $Z$ are $n_{r} \times k$ and $n_{s} \times k$ matrices respectively and with $k \ll \min(n_{r},n_{s})$. Using the low rank representation of $C$ we can now prove the following result.
\begin{lem}
Let $R$ be an $n_{r} \times n_{s}$ matrix with columns $R_{j}\ j=1,\ldots,n_{s}$ and let $X$ and $Z$ be an $n_{r} \times k$ and $n_{s} \times k $ matrices respectively, with columns $X_{j}$ and $Z_{j}\ \ j=1,\ldots k$. Finally, let $w$ be a vector of size $n_{s}$, then
\begin{equation} \label{eq:lemma}
(C \odot R) w = \sum_{j=1}^{k} X_{j}   \odot R (Z_{j} \odot w).
\end{equation}
\end{lem}
\begin{proof}
Statement~\eqref{eq:lemma} can be shown by simple transformations
\begin{align*}
(C \odot R) w &= (XZ^{\top} \odot R) w = \sum_{i=1}^{n_{r}}  w_{i} R_{i} \odot \sum_{j=1}^{k} X_{j} Z_{ji} = \sum_{j=1}^{k}  \sum_{i=1}^{n_{r}}  w_{i} R_{i} \odot  X_{j} Z_{ji}\\
& = \sum_{j=1}^{k} \left( X_{j}  \odot \sum_{i=1}^{n} R_{i} Z_{ji} w_{i} \right) = \sum_{j=1}^{k} X_{j}   \odot R (Z_{j} \odot w).
\end{align*}
\end{proof}

This identity implies that it is possible to obtain a stochastic representation to the misfit of the form
\begin{equation} \label{eq:optLR}
 \hf \|C \odot (P^{\top}A(u)^{-1}Q - D)\|_{\rm F}^{2}  = \hf \,
\bbE_{w}\left(\left\|\sum_{j=1}^{k} X_{j} \odot (P^{\top}A(u)^{-1}Q - D) (Z_{j} \odot w) \right\|^{2} \right).
\end{equation}
By using the identity~\eqref{eq:optLR} it is possible to obtain a stochastic approximation that extends the case
$C = \sigma^{-1} E$ using $k$ matrix solves. Here, $k$ is the rank of $C$  and we solve the stochastic programming problem
\begin{equation}
\label{clropt}
\min_{u} J(u) = \hf\, \bbE_{w}\left( \left\|\sum_{j=1}^{k} X_{j} \odot  (P^{\top}A(u)^{-1}Q (Z_{j} \odot w)  - D (Z_{j} \odot w) ) \right\|^{2} \right) + S(u).
\end{equation}

In other cases, where $C$ is not low rank, it is possible to approximate $C$ by a low rank matrix and use the above decomposition as an efficient alternative to the original misfit function~\eqref{eq:opt}. The advantage of this approach is that it has the same structure of the efficient algorithm used in previous work on simultaneous sources~\cite{HaberChungHerrmann2011}. Its disadvantage is that if the rank of $C$ is not low $k \not\ll \min(n_{r},n_{s})$ then this approach looses its computational advantages over~\eqref{eq:opt}.

\subsection{Stochastic Approximation of the data matrix}
\label{meth3}
For problems where the rank of $C$ is large and matrix completion fails, it is possible to obtain a different stochastic process that approximates the misfit function. It is straight forward to verify that
\begin{align*}
\|C \odot (P^{\top}A(u)^{-1}Q - D)\|_{\rm F}^{2} &= \|\, \bbE_{w} \left( C \odot \left( (P^{\top}A(u)^{-1}Q - D) ww^{\top} \right) \right) \|_{\rm F}^{2}\\
& = \| \, \bbE_{w} \left( C \odot ( (P^{\top}A(u)^{-1}Qw - Dw) w^{\top} ) \right) \|_{\rm F}^{2}.
\end{align*}
Using this approximation we can replace the deterministic problem~\eqref{eq:opt} with the stochastic programming problem
\begin{equation}
\label{absExp}
\min_{u} J(u) =  \hf \|  \bbE_{w} \left( C \odot ( (P^{\top}A(u)^{-1}Qw - Dw) w^{\top} ) \right) \|_{\rm F}^{2} + \alpha S(u).
\end{equation}

It is important to note that the problem~\eqref{absExp} is very different than the previous stochastic programming problems. The main difference is that the expected value is {\em inside} the norm rather than outside. Such stochastic programming problems are less common. However, it is possible to prove under certain conditions, that stochastic programming techniques converge to an optimal solution of the problem (see~\cite[Ch.~4]{shapiroBook}).

\subsection{Working with a subset of sources}
\label{meth4}
Finally, we consider a simple method that originated from a randomized Kaczmarz method~\cite{strohmer2009randomized}, which solves large linear systems and is used for instance in methods for  tomography such as OSEM~\cite{HudsonLarkin1994}. Let $j \in [1,\ldots,n_{s}]$ be an equally distributed random variable. Then, the original problem~\eqref{eq:opt} is equivalent to the following stochastic programming problem
\begin{align} \label{eq:optj}
	\begin{split}
\min_{u} \ J(u) &= {\frac 1{n_{s}}} \|C \odot (P^{\top}A(u)^{-1}Q - D)\|_{\rm F}^{2} + \alpha S(u) \\
&= \bbE_{j}\left( \|C_{j} \odot (P^{\top}A(u)^{-1}Q_{j} - D_{j})\|^{2}\right) + \alpha S(u),
\end{split}
\end{align}
where the expected value is on the index of the source $j$. This approach was proposed in \cite{KeesAscher2011} for the solution of the DC Resistivity problem using level sets. The approach is easy to implement, however, it may converge rather slowly if sources are not combined, especially if each source is sensitive only to a small part of the domain.

\section{Solving the optimization problem using stochastic programming techniques}
\label{sec3}

In this section we shortly discuss the solution of the stochastic programming problems~\eqref{eq:optj},~\eqref{absExp},~\eqref{clropt}, and~\eqref{stoproComp}. All proposed reformulations (besides~\eqref{absExp}) have the structure
\begin{equation} \label{eq:bayes}
\widehat u_{\alpha} = 	{\rm arg}\min_{u}\  {\cal J}(u) = \hf\ \bbE_{w} \left( \|f(u,w)\|^{2} \right) + \alpha S(u)
\end{equation}
where $f(u,w)$ is the misfit function for a single realization.
We use variants of the Sample Average Approximation (SAA) method~\cite{LinderothShapiroWright2006}, which we adopt for the problems at hand. The idea is to replace the expected value with a sum minimizing
\begin{equation} \label{eq:empirical}
\widehat u_{\alpha,N} = 	{\rm arg}\min_{u}\ {\cal J}_{N}(u) = {\frac 1 {2N}} \sum_{j = 1}^N \, \|f(u,w_{j})\|^{2} + \alpha S(u).
\end{equation}
Notice that in statistical methods problem~\eqref{eq:bayes} can be viewed as a \emph{risk minimization problem} and~\eqref{eq:empirical} as its \emph{empirical} pendant.
Obviously, an empirical approach is useful only if the number of samples $N \ll n_{s}$, which implies that the number of matrix inversions in the stochastic programming implementation is substantially smaller than the number of matrix inversions needed to solve the deterministic problem~\eqref{eq:opt}.

\emph{Stochastic Approximation} (SA) methods are another class to target stochastic optimization problem such as problem~\eqref{eq:bayes}, see~\cite{JuditskyLanNemirovskiShapiro2009}. However, here we prefer to use SAA methods for the following two reasons.
\begin{itemize}
\item First, SAA methods can be used with optimization methods that utilize curvature information such as Gauss-Newton algorithms. Contrary, SA methods are limited to first order methods, since to the best of our knowledge no convergence theory exists for higher order methods. Thus, even if the number of realizations for SAA methods is larger, the solution can be obtained in much fewer iterations when utilizing higher order methods. In this work we use the Gauss-Newton method that is know to have faster convergence compared to steepest descent methods \cite{pratt1999,hao}.
\item Since the results obtained by SAA methods are unbiased (or have a very small Bias), it is simple to use SAA in parallel architecture, with minimal communication and averaging the results at the end. This is unlike SA methods when parallelization requires communication is needed after each iteration.
\end{itemize}

\bigskip

Notice that the regularization parameter $\alpha$, which is often crucial for the inverse solution, is a-priori unknown. To overcome this problem we use a process of continuation of the regularization
parameter \cite{hao}. We start with a large value of $\alpha$, solve the optimization problem and continue to reduce $\alpha$ until the misfit reaches a desired value.

The process of reducing $\alpha$ can be done independent of choosing the sample size. However, this, in general, is inefficient. If the sample size is too small then reducing $\alpha$ does not lead to an approximate inverse solution. On the other hand, if the sample size is large then the solution of each optimization problem for each $\alpha$ is expensive.

We therefore use another ``enhancement'' for the SAA technique, that we found to be efficient. We combine the process of reducing the value of $\alpha$ with an adaptive approach that determines the sample size by continuation. That is, we start with a very small sample size and a large $\alpha$ and solve the optimization problem~\eqref{eq:empirical}. We then reduce $\alpha$ and increase the sample size {\em simultaneously}.
For each optimization problem we use ``hot starts''. That is, we start from the solution of the previous optimization problem. The process is terminated when the inverse solution obtains the target misfit {\em and} changes little while changing the sample size. This method is summarized in Algorithm~1.
\begin{algorithm}
\caption{Stochastic Programming for Inverse Problems}
\begin{algorithmic}[1]
\STATE Choose parameters $\gamma, \tau \in (0,1)$, $\alpha_0 \gg 0$ and $ \beta > 1$
\STATE Initialize, set $N_{1}=1, k=1$, and $\widehat u^0_{N_0,\alpha_0} = u_{0}$
\WHILE {not converged}
\STATE Choose $w_{[1,N_{k}]}$
\STATE{Solve optimization problem~\eqref{eq:empirical} for $\widehat u^{k}_{N_{k},\alpha_{k}}$
with initial guess $\widehat u^{k-1}_{N_{k-1},\alpha_{k-1}}$ }
\IF{  ${\frac 1 {2N_{k}}} \sum_{j = 1}^{N_{k}} \, \|f(u,w_{j})\|^{2} \ge {\rm tol}$}
\STATE{ $\alpha_{k+1} = \gamma \alpha_{k}$}
\ENDIF
\IF{  $\|\widehat u^{k}_{N_{k},\alpha_{k}} - \widehat u^{k-1}_{N_{k-1},\alpha_{k-1}}\|_{\infty} \le \tau $}
\STATE{ break}
\ENDIF
\STATE{Set $N_{k+1} = {\sf ceil}(\beta N_{k})$}
\STATE $k = k+1$
 \ENDWHILE
 \end{algorithmic}
\label{al:al1}
\end{algorithm}

Obviously, the foremost computational costs of the algorithm is in step~5, where we solve an optimization problem for $\widehat u_{\alpha,N}$. Since we are using the Gauss-Newton method, we found that only a few iterations (1-3) are required in order to converge. Moreover, the number of iterations computed with the largest sample size is very small.

\section{Numerical comparisons}
\label{sec4}

In this section we perform a numerical comparisons between the approaches discussed earlier in Section~\ref{sec2}.

\paragraph{Model Problem.} As a model problem we consider the DC Resistivity inversion, where the forward problem is a discretization of the static Maxwell's equations (see \cite{hao2,na})
\begin{eqnarray}
\label{maxwell}
\div  u \ \grad y_j =  q_j, \quad \mbox{for } \ y_j \in \Omega \quad \mbox{ and } \quad \nn \cdot \grad y_j = 0 \quad \mbox{for } \ y_j \in \partial \Omega \quad j=1,\ldots,n_s.
\end{eqnarray}
This is a common problem in geophysical exploration \cite{wardhow}. The goal is to recover the conductivity $u$ from measurements of the electric field $\grad y$.
%\tia{Eldad, more explanation on the model equations?}

\paragraph{Experimental settings.}

We use the SEG salt model which is a common model for benchmarking geophysical inverse problems \cite{Alkhalifah1998}. The experimental setting is shown in Figure~\ref{fig1}. The model consists of a large salt body with conductivity of $10^{-2}$\,S/m embedded within a layered media with conductivity that ranges from $1 \times 10^{-1}$ to $3 \times 10^{-1}$\,S/m.
\begin{figure}
\begin{center}
\includegraphics[width=10cm]{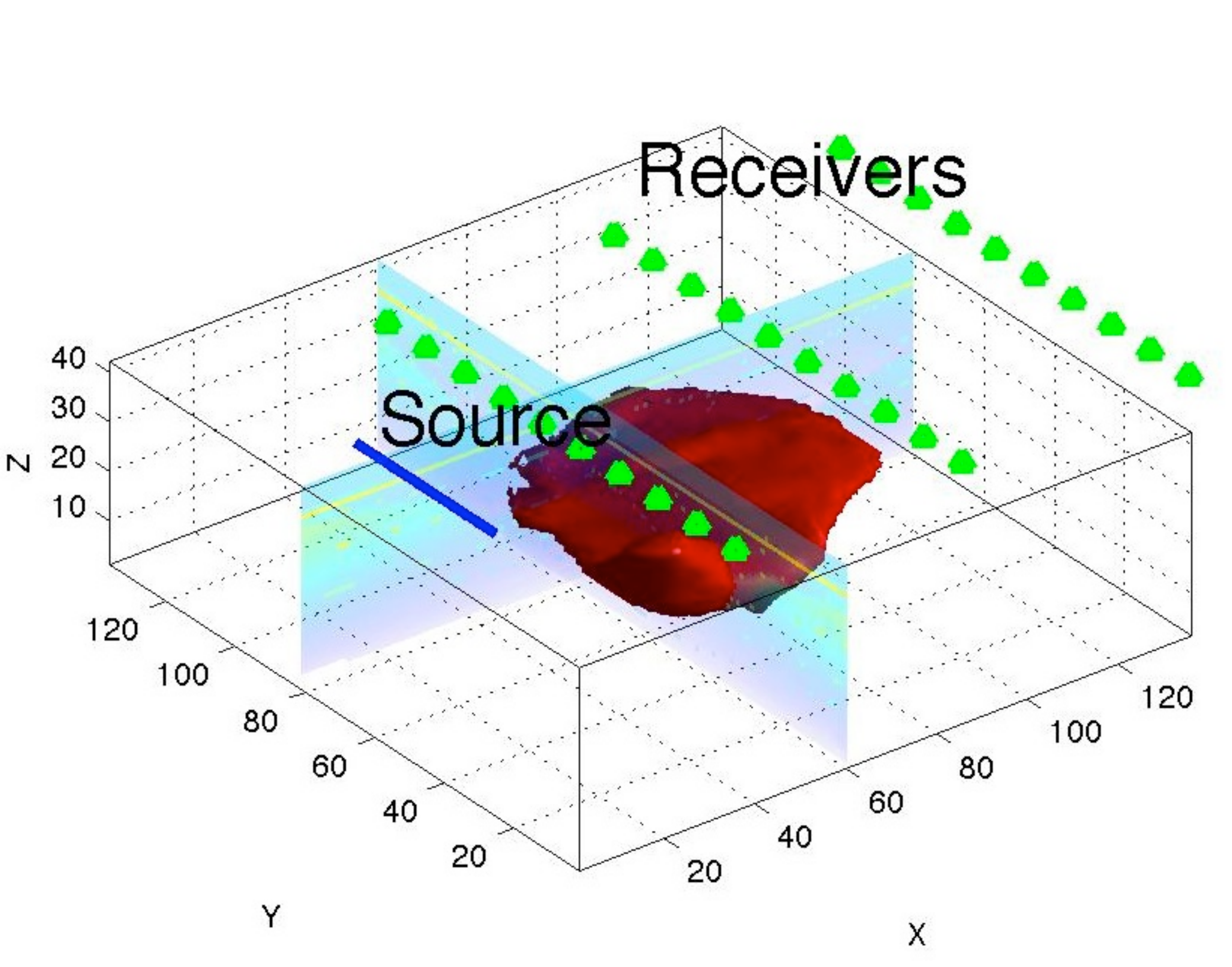}
\caption{The SEG model used for the numerical experiments and the experimental setting for a single source. Electric lines transmit current into the media and the resulting electric and or magnetic fields are recorded by the receivers.
\label{fig1}}
\end{center}
\end{figure}
The sources are line sources and we assume to have 800 dipole line sources, $20 \times 20 =400$ equally spaced that are pointing in the $x$-direction and the same number of  dipoles pointing in the $y$-direction. The receivers are assumed to measure 2 components of the electric field on the surface. We assume to have $900$ receivers spread over the surface uniformly. For this setting we have $900 \times 800 \times  2 = 1,440,000$ data points.

\paragraph{Discretization \& Forward Problem.} For the experiments we use a stretched mesh of $64 \times 64 \times 48$ cells where we assume that the conductivity of each cell is constant. We use a staggered grid (Yee's method) for the discretization of the forward problem in first order form \cite{yee,mm89,na,ha} which leads to the linear system
$$ A(u) = B^{\top}{\sf diag}(A_{v}u)B, $$
where $B$ is the discretization of the $\grad$-operator and $A_{v}$ is an averaging matrix. For the grid at hand, each forward problem (for a single source) consists of solving a linear system with $196,608$ unknowns. This forward problem is solved using the direct solver MUMPS, see \cite{MUMPS}. The computation of the forward problem in this environment takes about 1 minute. The matrix $A(u)$ in this application is symmetric positive semidefinite and since we are able to compute and store the Cholesky factors we reuse them for the computation of the adjoint problem and for the Gauss-Newton iteration.
%For details about the discretization and the solution technique of the forward problem
%see \cite{HoreshHaber2011}.
For the calculation of Jacobians and other related inverse problem quantities see \cite{hao2}.

\paragraph{Further Settings.} For the numerical experiments we set $\gamma  = 0.5$, $\tau = 10^{-2}$ and $\beta=2$.
We use Tikhonov regularization
$$ S(u) = u^{\top} L^{\top}Lu, $$
where we choose $L$ as a discretization of the corresponding gradient matrix.

\paragraph{Setup \& Computations.} To compare the computational costs of the different techniques we count the number of forward calculations. Since over $95\%$ of the computations are in this phase this is an accurate indicator for the overall computational costs of each algorithm. We perform~5 types of experiments corresponding to $5$ different types of the matrix $C$. In the first case we choose $C=\sigma^{-1}E$. This is the simplest case, where no data is missing and all data share the same standard deviation. The value of $\sigma$ is chosen as $1\%$ of the average value of the data. In the second experiment we change the data's standard deviation to be $1\%$ of each datum.
In the third, forth, and fifth experiments we use $70\%$, $40\%$, and $10\%$ of the data, randomly selected, with standard deviations
that are $1\%$ of each datum.

More details about the numerical experiments are summarized below.
\begin{itemize}
\item The Gauss-Newton method uses a conjugate gradient (CG) solver that iterates up to $5$ iterations. Each iteration involves calculating the forward and the adjoint problems.
\item The matrices $C$ that correspond to this experiment do not admit low rank and their singular values decay exponentially from $10^{3}$ to $10^{-1}$. We therefore used a space of $5$ and $10$ vectors for our calculations.
\item For the source-subset method we start with $10$ sources and increase the number of sources if needed. We have found to use up to $50$ sources at a time lead to satisfactory results.
\end{itemize}
\paragraph{Results.} The results of the numerical experiments are summarize in Table~\ref{tab1}. We record the total number of Gauss-Newton iterations as well as the number of forward problems solve. The recovery error (relative error) is also recorded and it is computed
as $\|u^{\rm computed} - u^{\rm true}\|/ \| u^{\rm true}\|$.

In Figure~\ref{fig2} we slice through the true solution and the recovered solutions obtained by the different methods.
\begin{figure}
\begin{center}
\begin{tabular}{cc}
 \includegraphics[width=8cm]{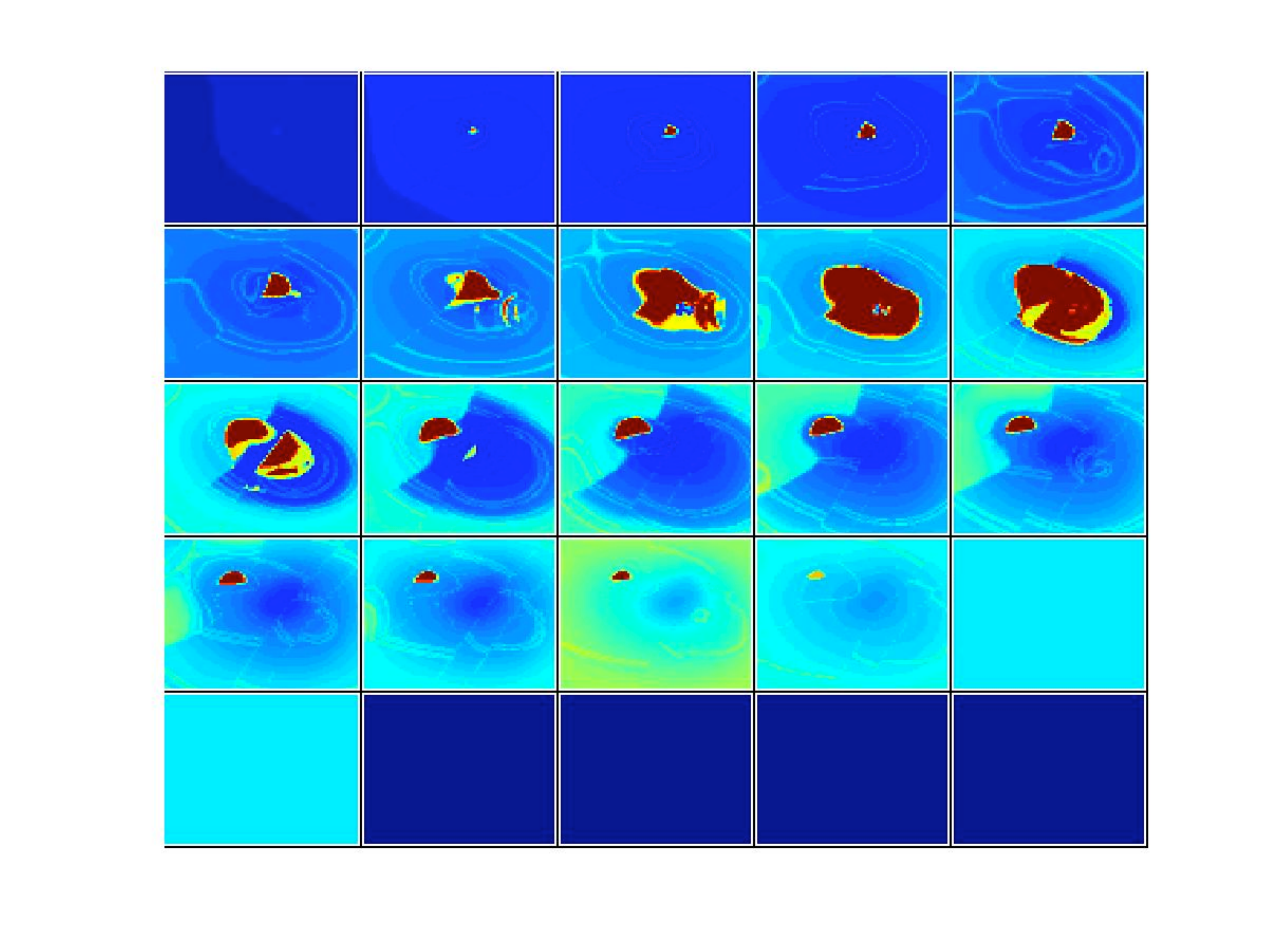} &
 \includegraphics[width=8cm]{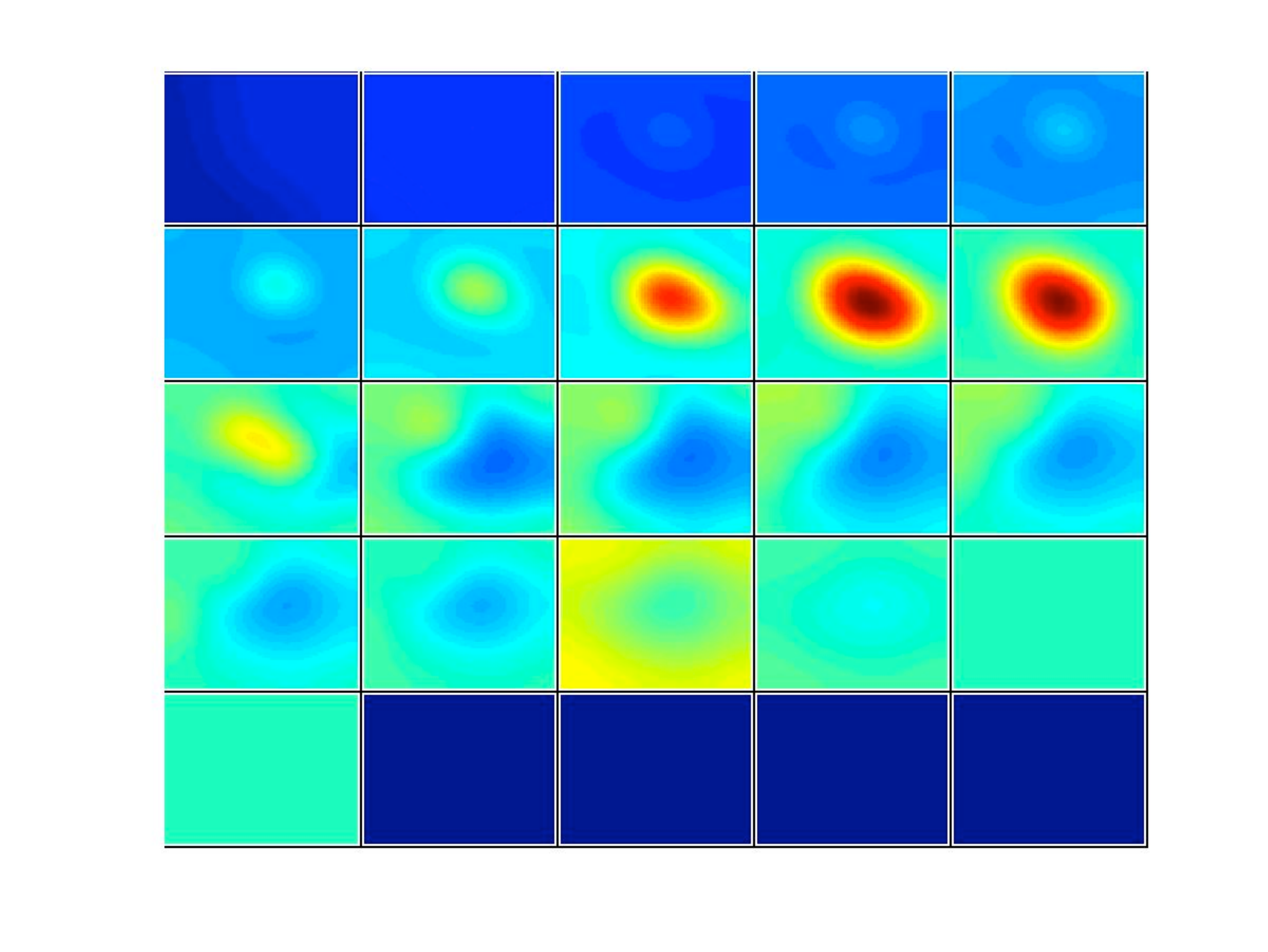} \\
(A) True model &  (B) Full data $\sigma$ uniform \\
\includegraphics[width=8cm]{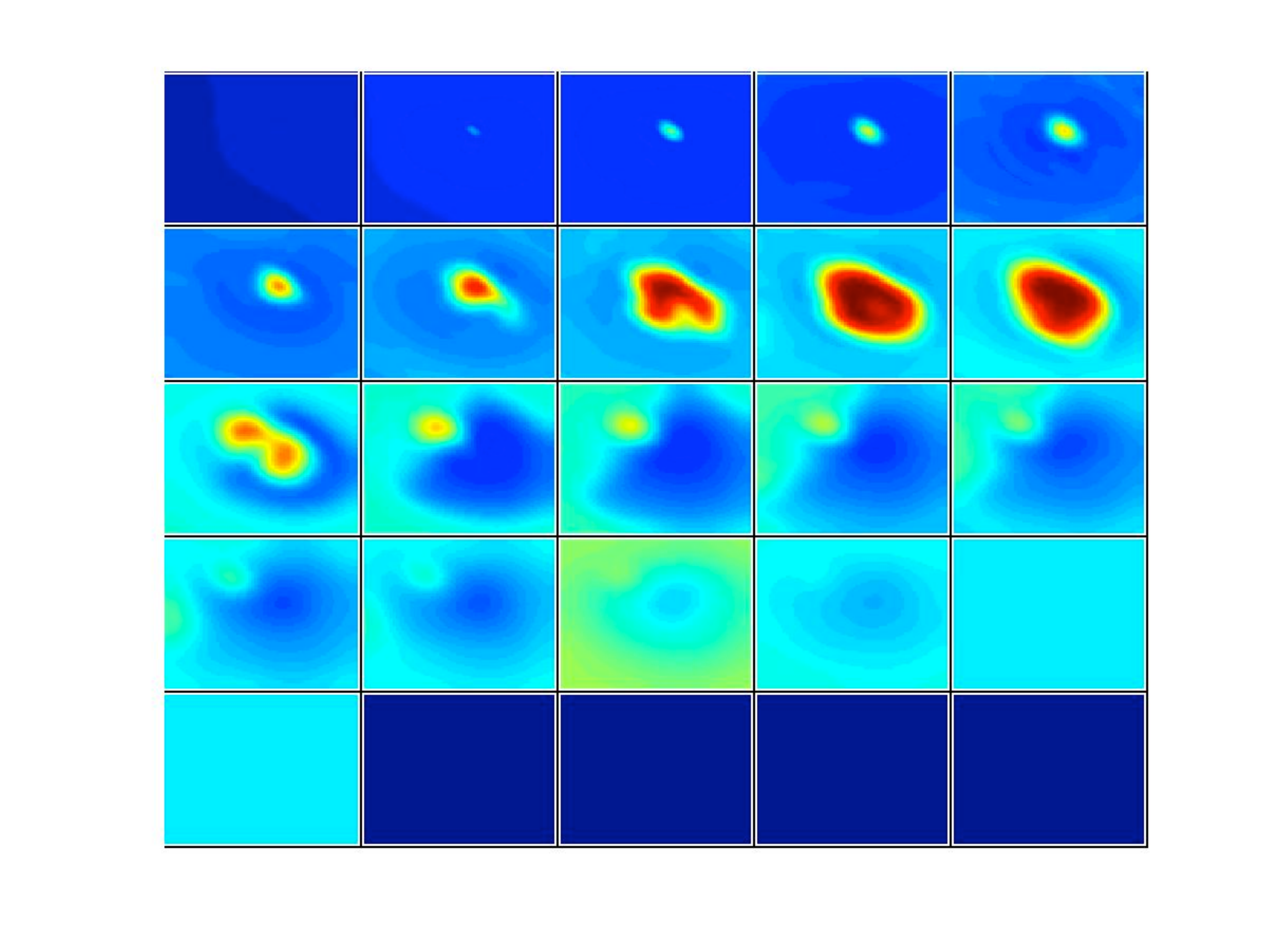} &
\includegraphics[width=8cm]{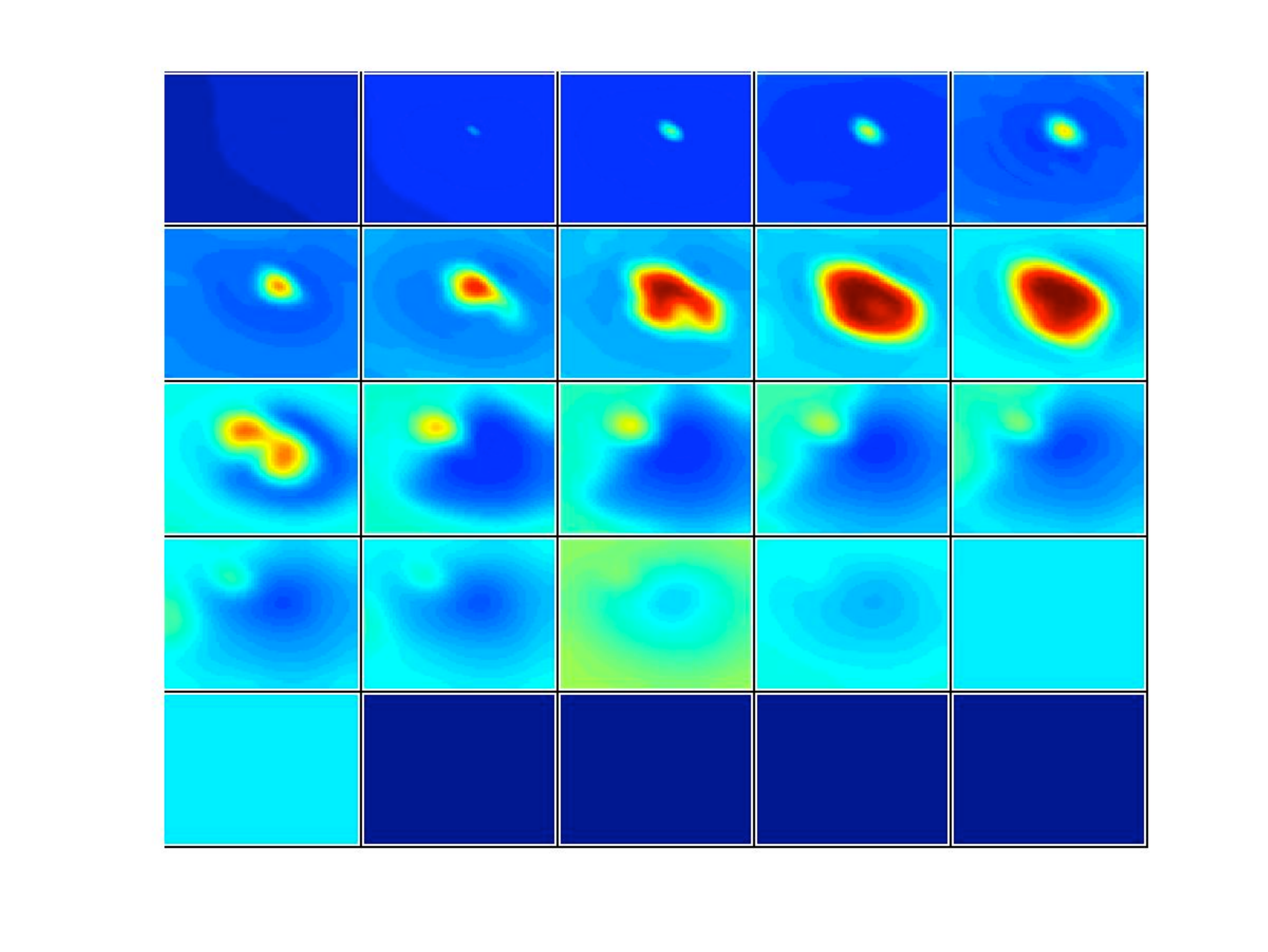}  \\
(C) Full data $\sigma$ non-uniform &  (D)  70\% data $\sigma$ non-uniform \\
\includegraphics[width=8cm]{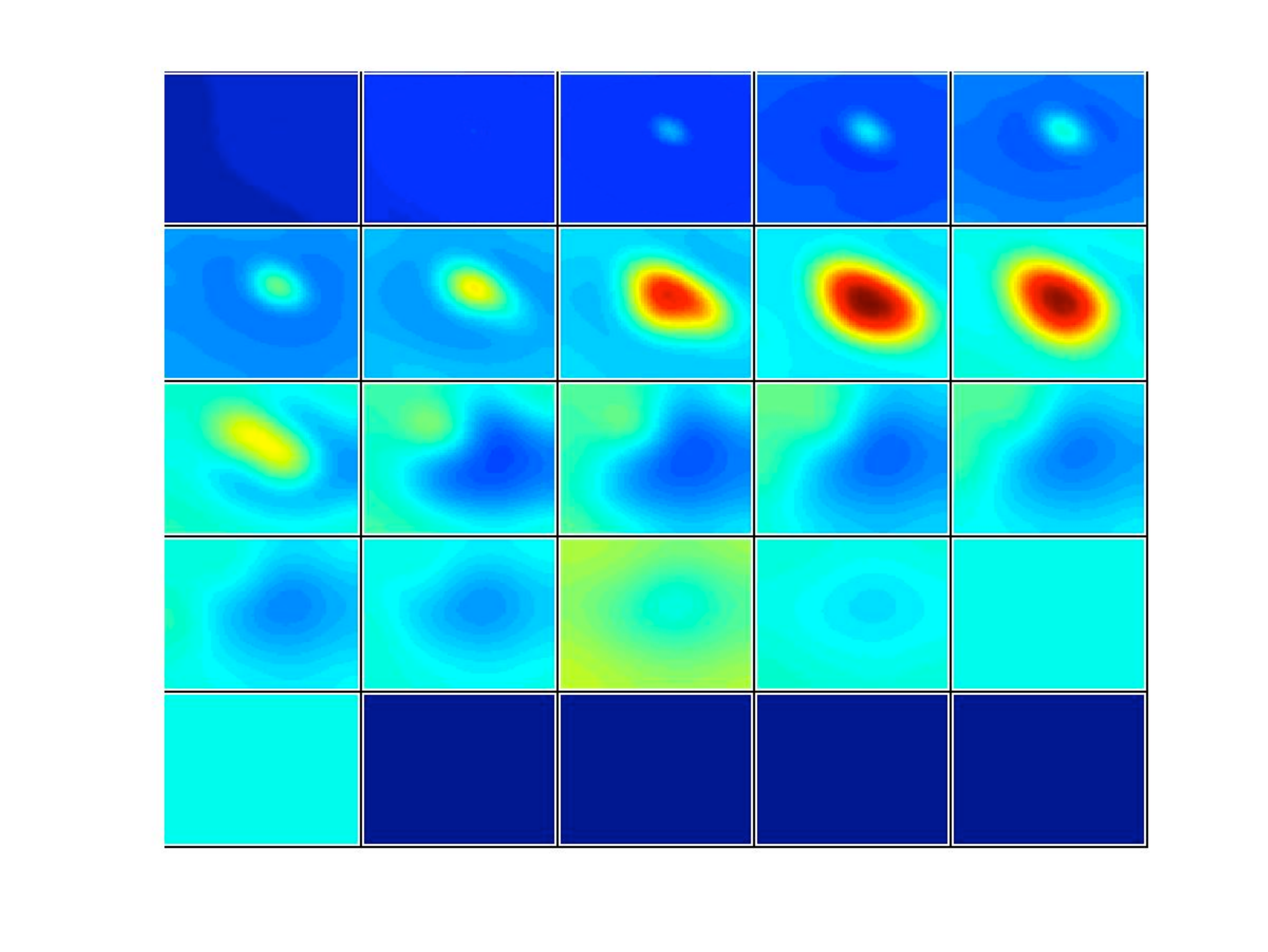} &
\includegraphics[width=8cm]{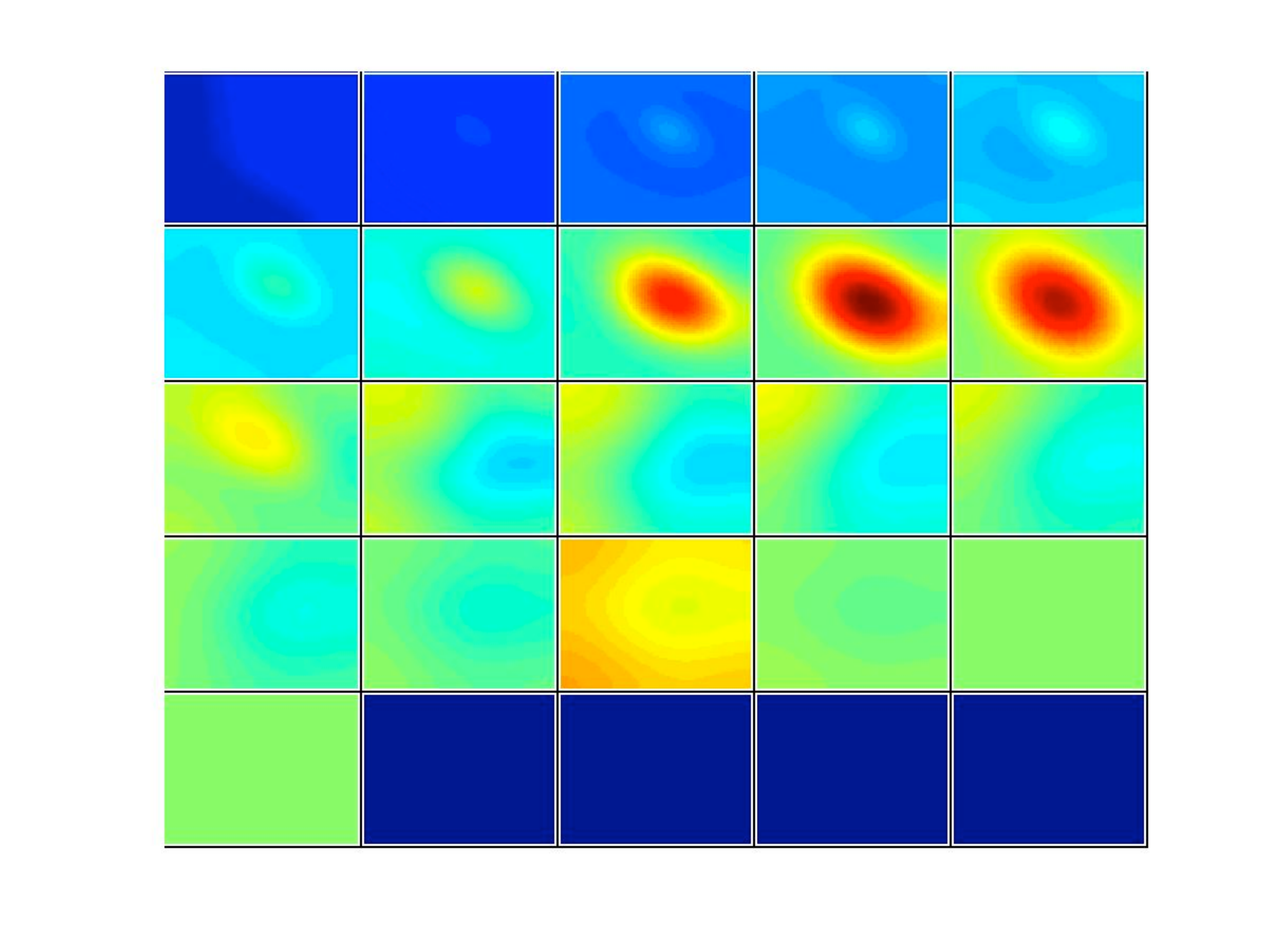}  \\
(E)  40\% data $\sigma$ non-uniform &  (F)  10\% data $\sigma$ non-uniform \\
\end{tabular}
\caption{The true model and the resulting inverted data obtained for different choices of $C$.
\label{fig2}}
\end{center}
\end{figure}

\begin{table}
\begin{center}
\begin{tabular}{|cccc|}
\hline
Method                     & GN Iterations & \# forward calculation & Recovery error \\
\hline
\multicolumn{4}{|c|}{ Uniform $C_{1} = \sigma^{-1} E$} \\
\hline
Low rank(1)            &         34                   &              2043                       &        0.28       \\
Stochastic approx  &         43                  &             9789                        &        0.29      \\
Subset                      &        54                   &            19446                        &       0.28        \\
\hline
\multicolumn{4}{|c|}{Standard deviation $(C_{2})_{ij} = \sigma_{ij}^{-1}  $} \\
\hline
Low rank(5)            &         33            &              11254                   &        0.19       \\
Low rank(10)          &         28           &              17232                    &        0.18       \\
Stochastic approx  &       56             &             20242                    &      0.19         \\
Subset                     &       55            &              21897                    &       0.18        \\
\hline
\multicolumn{4}{|c|}{Same as $C_{2}$ with $30\%$ of the data missing} \\
\hline
Data completion    &       35                &             2154                        &     0.29          \\
Low rank(5)            &        32            &              10913                   &        0.21       \\
Low rank(10)          &         27           &              16617                    &        0.20       \\
Stochastic approx  &       54             &             19519                    &      0.21         \\
Subset                     &       55            &              22011                    &       0.19        \\
\hline
\multicolumn{4}{|c|}{Same as $C_{2}$ with $60\%$ of the data missing} \\
\hline
Data completion    &       35                &             2142                        &     0.29          \\
Low rank(5)            &        31            &              10572                   &        0.24       \\
Low rank(10)          &         24           &              14771                    &        0.23       \\
Stochastic approx  &       54             &             19581                    &      0.25         \\
Subset                     &       55            &              22032                    &       0.22        \\
\hline
\multicolumn{4}{|c|}{Same as $C_{2}$ with $90\%$ of the data missing} \\
\hline
Data completion    &       35                &             2201                        &     0.31          \\
Low rank(5)            &        28            &               9548                   &        0.28       \\
Low rank(10)          &         21           &              12925                    &        0.28       \\
Stochastic approx  &       51             &             18493                    &      0.30         \\
Subset                     &       50            &              20029                    &       0.27        \\
\hline
\end{tabular}
\caption{Comparison between different recovery techniques. The data completion
correspond to Section~\ref{meth1}, the low rank corresponds to Section~\ref{meth2}, the
stochastic approximation corresponds to Section~\ref{meth3} and the subset corresponds to Section~\ref{meth4}.}
 \label{tab1}
\end{center}
\end{table}

The numerical experiments reveal some interesting observations.
\begin{itemize}
\item First, we note that the number of forward problems solved is roughly equal to 2 times the average number of realizations times the number of Gauss-Newton iterations times the average of inner CG iterations plus one,
\begin{multline*}
	\# {\rm forward\ problems} \approx \\ 2 (\# {\rm average\ realizations})
\left( (\# {\rm GN iter}) \times (\# {\rm average\ CG\ iter}) + 1\right).
\end{multline*}
Thus, reducing the number of realizations plays a crucial role in the computational cost.
\item If we assign a constant standard deviation to each datum, small  data in magnitude are not fitted well and as a result, the reconstruction has lower resolution (data not shown). This demonstrates the need for the development of the techniques suggested in this paper, see Figure~\ref{fig2}.
\item
Assigning each datum an appropriate standard deviation leads to the best results but tends
to also be the most expensive.
\item The  subset approach    yield the best recovery closely followed by the  low rank approach.
Nonetheless, the low-rank approach is substantially cheaper that the subset approach.
\item Reducing the number of data  effected the inversion however, the resulting models are
still very reasonable. This implies that the data admits redundancy. The question rises how to effectively
collect data that has less redundancies.
\item As the number of data is reduced, fewer iterations are needed. This is not surprising as it is
easier to fit the data when there are fewer measurements.
\end{itemize}

\section{Conclusions and Summary}
\label{sec5}

 In this paper we study the question of the so called simultaneous source with nonuniform standard deviation
 and with missing data. The technique offers significant computational saving over the traditional
 deterministic Gauss-Newton method. The saving is obtained by turning the problem into a stochastic
programming problem and then using a version of the Stochastic Average Approximation (SAA) to solve
the problem coupled with adaptive increase of the sample size. The paper study a number of different
techniques that allow for non-uniform standard deviation.

We have conducted extensive experiments on the DC resistivity inverse problem.
In terms of quality we have found that the nonlinear Kaczmartz iteration yields the best reconstruction
but with the highest cost. Similar results but with much lower cost are obtained by low-rank approximation to the variances matrix $C$. Data completion using a reduced model seem to also yield reasonable results.

The application of our approach to other inverse problems in geophysics and medical physics is straight
forward and can lead to significant saving in many applications.

\bibliographystyle{plain}
\bibliography{biblio}

\begin{thebibliography}{10}

\bibitem{Alkhalifah1998}
T.~Alkhalifah.
\newblock The fast marching method in spherical coordinates: Seg/eage saltdome
  model.
\newblock {\em SEP-97}, pages 251--264, 1998.

\bibitem{MUMPS}
P.R. Amestoy, I.S. Duff, J.-Y. L'Excellent, and J.~Koster.
\newblock A fully asynchronous multifrontal solver using distributed dynamic
  scheduling.
\newblock {\em SIAM Journal on Matrix Analysis and Applications}, 23:15--41,
  2001.

\bibitem{brocea}
L.~Borcea.
\newblock Electrical impedance tomography.
\newblock {\em Inverse Problems}, 18:99--136, 2002.

\bibitem{jc1}
J.~Claerbout and F.~Muir.
\newblock Robust modeling with erratic data.
\newblock {\em Geophysics}, 38:826--844, 1973.

\bibitem{devony}
A.~J. Devaney.
\newblock The limited-view problem in diffraction tomography.
\newblock {\em Inverse Problems}, 5:510--523, 1989.

\bibitem{ha}
E.~Haber and U.~Ascher.
\newblock Fast finite volume simulation of 3{D} electromagnetic problems with
  highly discontinuous coefficients.
\newblock {\em SIAM Journal on Scientific Computing}, 22:1943--1961, 2001.

\bibitem{hao}
E.~Haber, U.~Ascher, and D.~Oldenburg.
\newblock On optimization techniques for solving nonlinear inverse problems.
\newblock {\em Inverse Problems}, 16:1263--1280, 2000.

\bibitem{hao2}
E.~Haber, U.~Ascher, and D.~Oldenburg.
\newblock Inversion of {3D} electromagnetic data in frequency and time domain
  using an inexact all-at-once approach.
\newblock {\em Geophysics}, 69:1216--1228, 2004.

\bibitem{HaberChungHerrmann2011}
E.~Haber, M.~Chung, and F.~J. Herrmann.
\newblock An effective method for parameter estimation with pde constraints
  with multiple right hand sides.
\newblock {\em SIAM Journal on Optimization}, 22(3):739--757, 2012.

\bibitem{HudsonLarkin1994}
H.M. Hudson and R.S. Larkin.
\newblock Accelerated image reconstruction using ordered subsets of projection
  data.
\newblock {\em IEEE Transactions on Medical Imaging}, 13(4):601--609, 1994.

\bibitem{JuditskyLanNemirovskiShapiro2009}
A.~Juditsky, G.~Lan, A.~Nemirovski, and A.~Shapiro.
\newblock Stochastic approximation approach to stochastic programming.
\newblock {\em SIAM Journal on Optimization}, 19, 2009.

\bibitem{krebs09ffw}
J.R. Krebs, J.E. Anderson, D.~Hinkley, R.~Neelamani, S.~Lee, A.~Baumstein, and
  M.-D. Lacasse.
\newblock Fast full-wavefield seismic inversion using encoded sources.
\newblock {\em Geophysics}, 74(6):WCC177--WCC188, 2009.

\bibitem{LinderothShapiroWright2006}
J.J. Linderoth, A.~Shapiro, and S.~Wright.
\newblock The empirical behavior of sampling methods for stochastic
  programming.
\newblock {\em Annals of Operations Research}, 142:215--241, 2006.

\bibitem{mm89}
T.~Madden and R.~Mackie.
\newblock Three-dimensional magnetotelluric modeling and inversion.
\newblock {\em Proceedings of the IEEE}, 77:318--321, 1989.

\bibitem{neelamani08dos}
N.~Neelamani, C.~Krohn, J.~Krebs, M.~Deffenbaugh, and J.~Romberg.
\newblock Efficient seismic forward modeling using simultaneous random sources
  and sparsity.
\newblock In {\em SEG International Exposition and 78th Annual Meeting}, 2008.

\bibitem{na1}
G.~Newman and D.~Alumbaugh.
\newblock Three-dimensional massively parallel electromagnetic inversion --{I.
  T}heory.
\newblock {\em Geophysical Journal International}, 128:345--354, 1997.

\bibitem{na}
G.A. Newman and D.L. Alumbaugh.
\newblock Frequency-domain modelling of airborne electromagnetic responses
  using staggered finite differences.
\newblock {\em Geophysical Prospecting}, 43:1021--1042, 1995.

\bibitem{parker}
R.~L. Parker.
\newblock {\em Geophysical Inverse Theory}.
\newblock Princeton University Press, Princeton NJ, 1994.

\bibitem{pratt1999}
R.G Pratt.
\newblock Seismic waveform inversion in the frequency domain, part 1: Theory,
  and verification in a physical scale model.
\newblock {\em Geophysics}, 64:888--901, 1999.

\bibitem{RohmbergGeop2010}
J.~Romberg, R.~Neelamani, C.~Krohn, J.~Krebs, M.~Deffenbaugh, and J.~Anderson.
\newblock Efficient seismic forward modeling and acquisition using simultaneous
  random sources and sparsity.
\newblock {\em Geophysics}, 75(6):WB15--WB27, 2010.

\bibitem{shapiroBook}
A.~Shapiro, D.~Dentcheva, and D.~Ruszczynski.
\newblock {\em Lectures on Stochastic Programming: Modeling and Theory}.
\newblock SIAM, Philadelphia, 2009.

\bibitem{smvoz}
N.C. Smith and K.~Vozoff.
\newblock Two dimensional {DC} resistivity inversion for dipole dipole data.
\newblock {\em IEEE Transaction on Geoscience and Remote Sensing, Special Issue
  on Electromagnetic Methods in Applied Geophysics}, GE 22:21--28, 1984.

\bibitem{strohmer2009randomized}
T.~Strohmer and R.~Vershynin.
\newblock {A randomized Kaczmarz algorithm with exponential convergence}.
\newblock {\em Journal of Fourier Analysis and Applications}, 15(2):262--278,
  2009.

\bibitem{taran}
A.~Tarantola.
\newblock {\em Inverse Problem Theory}.
\newblock Elsevier, Amsterdam, 1987.

\bibitem{TradUlrychSacchi2001}
D.~Trad, T.J. Ulrych, and M.D. Sacchi.
\newblock Accurate interpolation with high resolution time variant radon
  transforms.
\newblock {\em Geophysics}, 67:644--656, 2001.

\bibitem{KeesAscher2011}
K.~van~den Doel and U.~M. Ascher.
\newblock Adaptive and stochastic algorithms for {EIT} and {DC} resistivity
  problems with piecewise constant solutions and many measurements.
\newblock Technical report, University of British Columbia, 2011.

\bibitem{LeeuwenAravkinHerrmann2011}
T.~van Leeuwen, A.~Aravkin, and F.~J. Herrmann.
\newblock Seismic waveform inversion by stochastic optimization.
\newblock {\em International Journal of Geophysics}, 2011(ID 689041):1--18,
  2011.

\bibitem{wardhow}
S.H. Ward and G.W. Hohmann.
\newblock Electromagnetic theory for geophysical applications.
\newblock {\em Electromagnetic Methods in Applied Geophysics}, 1:131--311,
  1988.
\newblock Soc. Expl. Geophys.

\bibitem{yee}
K.S. Yee.
\newblock Numerical solution of initial boundary value problems involving
  {M}axwell's equations in isotropic media.
\newblock {\em IEEE Transactions on Antennas and Propagation}, 14:302--307,
  1966.

\end{thebibliography}

\end{document}